\documentclass[a4paper, 12pt]{article} 
\usepackage[right=1in,left=1in,top=1in,bottom=1in]{geometry}
\usepackage{amsfonts, amssymb, epic, ecltree}
\usepackage{hyperref}
\usepackage{amsmath}
\usepackage[utf8]{inputenc}
\usepackage[T1]{fontenc}
\usepackage{lmodern}
\usepackage{pxfonts}
\usepackage{tikz}   
\usepackage{tikz-3dplot}
\usepackage{tcolorbox}
\usepackage{pgfplots}
\usepackage{setspace}
\usepackage[titles]{tocloft}

\usepackage{amsthm}
\usepackage{stmaryrd}
\usepackage{calrsfs}
\usepackage{graphicx}
\DeclareMathAlphabet{\pazocal}{OMS}{zplm}{m}{n}
\newtheorem{theorem}{Theorem}
\newtheorem{corollary}{Corollary}
\newtheorem{lemma}{Lemma}

\theoremstyle{definition}
\newtheorem{definition}{Definition}
\theoremstyle{remark}

\title{\Large{The Distribution Postulate in Algorithmic Bohmian Mechanics}}
\author{Jeffrey A. Barrett\footnote{Department of Logic and Philosophy of Science, University of California, Irvine, Irvine, CA 92697-5100. Email: j.barrett@uci.edu} 
, Eddy Keming Chen\thanks{Department of Philosophy,  University of California, San Diego, 9500 Gilman Dr, La Jolla, CA 92093-0119. Email: eddykemingchen@ucsd.edu  }
~~and 
Josiah Lopez-Wild\thanks{Department of Logic and Philosophy of Science, University of California, Irvine, Irvine, CA 92697-5100. Email: jlopezwi@uci.edu}}

\date{ }

\begin{document}

\maketitle

\begin{abstract}
In order to make the right empirical predictions Bohmian mechanics requires a special statistical boundary condition---the distribution postulate---but it is unclear how best to understand this condition. We show how one might use the theory of algorithmic randomness to formulate the distribution postulate as an \emph{objective constraining law}. The framework requires us to say something about admissible quantum-mechanical states and measurements. In return, algorithmic Bohmian mechanics (aBM) guarantees the standard Born statistics for a collection of canonical quantum experiments in the limit, not just with high probability. The algorithmic distribution postulate provides a sharp typicality condition, clarifies the status of quantum probabilities in the deterministic theory, and provides a concrete example of how notions provided by the theory of algorithmic randomness can aid in specifying the content of a physical law.
\vspace{5mm}

\end{abstract}


\section{Introduction}

Bohmian mechanics is a deterministic theory that makes the same empirical predictions as the standard collapse formulation of quantum mechanics for particle positions. As there are no stochastic processes, forward-looking probabilities in the theory are usually understood as epistemic, reflecting our ignorance of the particle positions. As with classical statistical mechanics, probabilities in Bohmian mechanics can arise from a special statistical boundary condition. In Bohmian mechanics this condition is called the \textit{distribution postulate}.

The theory of algorithmic randomness can be used to give a clear statement of the distribution postulate as an objective constraining law. To do this, one also needs to say something regarding admissible quantum states and measurements. We will call the full theory \textit{algorithmic Bohmian mechanics} (aBM).\footnote{While there are disanalogies between the two cases, a similar approach can be used to give the special boundary conditions (the past hypothesis and, more specifically, the statistical postulate) required for predictive compatibility between classical statistical mechanics and thermodynamics.} 

To see how to use algorithmic randomness to characterize the distribution postulate in Bohmian mechanics, we will begin with a brief discussion of how classical Bohmian mechanics works, then turn to how the distribution postulate has been understood. We will then introduce the notions required to formulate algorithmic Bohmian mechanics.

\section{Bohmian mechanics and the distribution postulate}

In its simplest form, Bohmian mechanics is
characterized by the following principles:\footnote{This description of Bohm's theory follows Bell
(1987, 127) rather than Bohm's quantum-potential description.}

\begin{itemize}

\item[1.] State Description: The complete physical state at a time is given by the wave function $\psi$ and the determinate particle configuration $Q$.

\item [2.] Wave Dynamics: The time evolution of the wave function is given by the usual linear dynamics. In the simplest case, this is just
\begin{equation}
i \hbar \frac{\partial \psi}{\partial t} =  \hat{H} \psi
\end{equation}
More generally, one uses the form of the unitary dynamics appropriate to the system and observables at hand.

\item[3.] Particle Dynamics: The particles move according to
\begin{equation}
\frac{d Q_k}{dt} = \frac{1}{m_k} \frac{\mbox{Im}(\psi^* \nabla_k \psi)}{\psi^* \psi}\Big|_Q
\end{equation}
where $m_k$ is the mass of particle $k$ and $Q$ is the current configuration.

\item[4.] Distribution Postulate: At an initial time $t_0$ the probability density for the configuration $Q$ is given by $\rho(Q, t_0)= |\psi(Q, t_0)|^2$. 

\end{itemize}

If there are $N$ particles, then $\psi$ is a function in $3N$-dimensional configuration space (three dimensions for the position of each particle), and the current particle configuration $Q$ is represented by a single point in configuration space that determines the position of every particle.  Each particle moves in a way that
depends on its position, the evolution of the wave function, and the positions of every other particle. One might picture the particle dynamics described by rule~3 by imagining the point representing the $N$-particle configuration as being carried along by the probability currents generated by the linear evolution
of the wave function $\psi$ in configuration space as if it were a massless particle being pushed about by a compressible fluid. This probability is often given an epistemic interpretation, reflecting our ignorance of the initial particle positions. (In our approach, we will reformulate it as an objective constraining law, the algorithmic distribution postulate in \S8.)

Bohmian mechanics ensures that the particle configuration is always determinate. Given the wave function, the particle configuration in turn determines all measurement records. The result of an $x$-spin measurement on an electron might, for example, be recorded in the direction in which the electron is deflected by the inhomogeneous magnetic field of a Stern-Gerlach measuring device.

Since the total particle configuration can be thought of as being carried around by the probability
current in configuration space, the probability of the particle configuration being found in a particular region of configuration space changes as the integral of $|\psi|^2$ over that region changes.  More specifically, the continuity equation
\begin{equation}
\frac{\partial \rho}{\partial t} + \mbox{div}(\rho v^\psi) = 0
\end{equation}
is satisfied by the probability density $\rho=|\psi|^2$.  This means that if the epistemic probability density for the particle configuration is ever $|\psi|^2$, then it will always be $|\psi|^2$, unless one makes an observation. This is the equivariance of the Born measure under the Bohmian dynamics.\footnote{See D\"urr, D., S. Goldstein, and N. Zangh\'i (1993) for a discussion of
the equivariance of the statistical distribution $\rho$.} It follows from this that if one starts with a prior epistemic probability density of $\rho(t_0)=|\psi(t_0)|^2$, then one should update this probability density at time $t$ so that $\rho(t)=|\psi(t)|^2$.  And if one makes an observation, then the epistemic probability density will be given by the wave function conditional on one's measurement result.

The key idea here is that it follows from the theory's two dynamical laws that if the epistemic probability density is $\rho(t_0)=|\psi(t_0)|^2$ at time $t_0$, then the push-forward probability density at a later time $t$ is $\rho(t)=|\psi(t)|^2$. In this way, the distribution postulate delivers the standard forward-looking Born probabilities as epistemic probabilities regarding particle positions.

\section{An $x$-spin measurement in Bohmian mechanics}

A simple example will illustrate how measurement works in Bohmian mechanics. We will use this basic Stern-Gerlach setup in the more complicated experiments we consider later.
\begin{figure}[!ht]
  \centering
    \includegraphics[scale=.6]{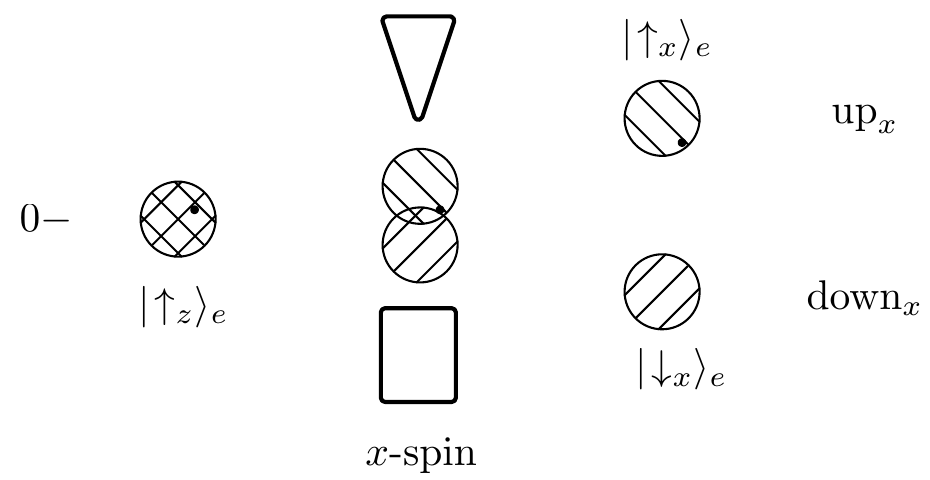}
    \caption{An $x$-spin measurement in Bohmian mechanics simple.}
\end{figure}

Consider a spin-$1/2$ particle $e$ in initial state
$$
|\!\uparrow_z\rangle_{e}|\psi\rangle_{e} = 1/\sqrt{2} (|\!\uparrow_x\rangle_{e} + |\!\downarrow_x\rangle_{e})|\psi\rangle_{e},
$$
and suppose that the probability density of the wave packet $|\psi\rangle_{e}$ is uniform and spherically symmetric.

Given the symmetry of the experiment, the wave packets
$|\!\uparrow_x\rangle_{e}|\psi_{\text{up}}\rangle_{e}$ and $|\!\downarrow_x\rangle_{e}|\psi_{\text{down}}\rangle_{e}$ are deflected up and down respectively by the Stern-Gerlach magnets, and if the particle $e$ starts in the top half of its wave function (as indicated in the figure), then it will be carried along by the probability currents, end up associated with the $|\!\uparrow_x\rangle_{e}|\psi_{\text{up}}\rangle_{e}$ wave packet, and hence deflected up. Had it been in the bottom-half of its wave packet, it would have been carried along by the probability currents, end up associated with the $|\!\downarrow_x\rangle_{e}|\psi_{\text{down}}\rangle_{e}$ wave packet, and hence deflected down.\footnote{See Goldstein (2025) for a description of Bohmian mechanics and Barrett (2019, 193-219) for a detailed discussion of such experiments.}

Bohmian mechanics ensures that the particle $e$'s position is always determinate. In this case, $e$'s position after the interaction with the magnets can be taken to constitute the $x$-spin measurement record, and its post-measurement position is, in turn, determined by whether $e$ was in the top or bottom half of the initial wave packet. So if the probability of the particle being in the top half and bottom half are both $1/2$, then the probability of each $x$-spin result is $1/2$.

The role of the distribution postulate in the theory is to provide prior probabilities for the configuration that yields posterior probabilities that agree with the standard quantum mechanical predictions given by the Born rule. If the initial configuration satisfies the distribution postulate, then one gets the standard forward-looking quantum probabilities for particle positions. So if all measurements are ultimately measurements of the position of something, then the theory makes the standard quantum predictions.

But how are we to understand the distribution postulate? Is it somehow entailed by basic principles of reason and the theory's dynamics? Or is it an independent law of nature? Is it just a recommendation for how a rational agent should assign prior credences to particle positions? Or is it a law that stipulates something regarding objective physical chance? And if so, then what exactly does the law say?\footnote{Each of the understandings listed here has been promoted at one time or another. See Callender (2007) for a survey. For the typicality reading, on which the distribution postulate expresses that the initial configuration is typical relative to the measure $|\psi|^2$, see D\"urr, Goldstein, and Zangh\'i (1992); for some recent discussions of the meaning of typicality in quantum mechanics and statistical mechanics, see Maudlin (2011), Allori (2020), Lazarovici (2023), and Wilhelm (2022). For alternative proposals that aim to derive quantum equilibrium dynamically rather than impose it as a boundary condition, see Valentini (1991) and Valentini and Westman (2005). For broader background on the interpretations of probability, see H\'ajek (2023).  }

David Albert characterized the distribution postulate as an objective physical law in characteristic Albert prose:
\begin{quote}
Suppose that the form of a certain single-particle wave function at a certain initial moment $t$ is $\psi(x,t)$; and suppose that at just that particular moment God places the \textit{particle} associated with that wave function $\ldots$ somewhere in the world, and suppose that God makes use of some genuinely \textit{probabilistic} procedure for deciding precisely where to \textit{put} that particle, and suppose that that procedure happens to entail that the probability that the particle gets put at any particular point is equal to the square of the absolute value of the particle's wave function at that point.'' (1992, 138)
\end{quote}
If the particle configuration is determined in this way at an initial time, then the Bohmian dynamics ensures that the epistemic probabilities one calculates for future times will be the standard quantum probabilities.

It is still unclear, however, precisely what the distribution postulate says. When and how does God determine the initial particle configuration? What is the ``genuinely probabilistic procedure'' she uses to decide where to put that particle? What precisely does it mean to have a genuinely probabilistic procedure? In short, where do the chances come from and in what sense are they objective?

The theory of algorithmic randomness allows one to give a clear formulation of the distribution postulate as an \textit{algorithmic constraining law}. The approach we will take here is based on the algorithmic notion of Martin-L\"{o}f randomness.\footnote{As described in Martin-L\"of (1966).}

\section{Algorithmic randomness}\label{sec:Algorithmic Randomness}

One can define what it means for a sequence to be (unbiased) Martin-L\"of random by considering the set of statistical tests that such a sequence will pass. A Martin-L\"of test is a sequence $\{U_n\}_{n\in \omega}$ of uniformly $\Sigma^0_1$ classes such that $\mu(U_n) \leq 2^{-n}$ for all $n$, where $\mu$ is the unbiased Lebesgue measure over the sequences. It is this uniform measure that makes the test unbiased. Being uniformly $\Sigma^0_1$ means that there is a single constructive specification of the sequence of classes. A constructive specification is one that can be represented by an ordinary algorithm.\footnote{See Dasgupta (2011), Li and Paul Vit\'anyi (2019), and Barrett and Huttegger (2021) for further details regarding algorithmic randomness.} The idea is that each sequence $\{U_n\}_{n\in \omega}$ of uniformly $\Sigma^0_1$ classes corresponds to a way that a sequence might be special and thus fail an associated statistical test for randomness. A sequence passes a particular Martin-L\"of test if it is not special in the specified computable sense checked by the test.

Let $2^\omega$ be the set of all $\omega$-length sequences (infinite-length sequences indexed by ordinal $\omega$). A class $C \subset 2^\omega$ is Martin-L\"of null if there is a Martin-L\"of test$\{U_n\}_{n\in \omega}$ such that $C \subseteq \bigcap_n U_n$. A sequence $S \in 2^\omega$ is \emph{Martin-L\"of random} if and only if $\{S\}$ is not Martin-L\"of null. That is, a sequence $S$ is Martin-L\"of random if and only if it passes every Martin-L\"of test. The idea is that an $\omega$-sequence~$\sigma$ is \emph{Martin-L\"{o}f random} if and only if it is not special in any computable sense.

It may be helpful to consider a few properties of an unbiased Martin-L\"{o}f random sequence. If a binary $\omega$-sequence~$\sigma$ is Martin-L\"{o}f random, then the limiting relative frequencies of $0$'s and $1$'s are each $1/2$, and $\sigma$ passes every effective test for statistical independence. This captures the intuition that a random sequence will have an unbiased limiting relative frequency and appear to be the result of a statistically independent process.  A sequence~$\sigma$ is Martin-L\"{o}f random if and only if there is a constant~$c$ such that all finite initial segments of the sequence are $c$-incompressible by a prefix-free Turing machine. This captures the intuition that one cannot effectively characterize a random sequence. A sequence is Martin-L\"{o}f random if and only if no constructive martingale succeeds on it. This captures the intuition that no computable fair betting strategy will do better than chance. And since almost all $\omega$-sequences are Martin-L\"{o}f random in Lebesgue measure, this notion also accords with the intuition that random sequences are typical. In sum, Martin-L\"{o}f randomness provides precisely what one should want from a clearly specified notion of randomness or typicality relative to an unbiased measure. And as we will see, Martin-L\"{o}f randomness can be extended to provide notion of typicality relative to any computable measure.

\section{Two tentative constructions}

We want to use an algorithmic notion like Martin-L\"{o}f randomness to specify the distribution postulate. In the present section, we will consider two tentative constructions to illustrate the idea, discuss why they falls short of what one might want, then say something about the relationship between statistics and probability. In the remainder of the paper, we will show that aBM predicts random sequences of measurement results that exhibit the standard quantum relative frequencies for experiments like those described in this section.

Consider an infinite sequence of spin-$1/2$ particles $S_1, S_2, \ldots S_k, \ldots$ each in an eigenstate of $z$-spin and associated with a spherically symmetric wave packet $|\psi\rangle$ of uniform probability density. We will write the state of particle $S_k$ as
$$
|\!\uparrow_z\rangle_{S_k}|\psi\rangle_{S_k} = 1/\sqrt{2} (|\!\uparrow_x\rangle_{S_k} + |\!\downarrow_x\rangle_{S_k})|\psi\rangle_{S_k}.
$$
Suppose that one plans to measure the $x$-spin of each particle in turn by  symmetrically separating the $|\!\uparrow_x\rangle$ and $|\!\downarrow_x\rangle$ components of each wave packet in the $x$-direction, recording~$0$ for~$\downarrow_x$ and~$1$ for~$\uparrow_x$, thereby producing an $\omega$-sequence of results.

According to Bohmian mechanics, the results of each measurement will depend on the position of each particle in its wave packet in the $x$-direction. As we saw earlier, if it is in the top half of its wave packet, then it will be deflected in the $x$-up direction and end up associated with an $x$-spin-up flavored wave packet; and if it is in the bottom half, it will be deflected in the $x$-down direction and end up associated with an $x$-spin-down flavored wave packet.

The idea is to say where the particles start by appeal to an algorithmic law. Consider a law $L^\star$ that stipulates that the initial positions of the particles were determined by the terms of a Martin-Löf random binary sequence $\sigma_{MLR}$: if the $k$th term of $\sigma_{MLR}$ is 1, then $S_k$ is in the top half of its wave packet; else it is in the bottom half. Such a distribution of particles relative to the Martin-L\"{o}f sequence $\sigma$ is illustrated in figure~2.

As a concrete example, consider an unbiased Martin-L\"{o}f random sequence
$$
\sigma_{MLR}=(1,0,1, \ldots)
$$

and the infinite series of wave packets represented in figure~2. 
\begin{figure}[!ht]
  \centering
    \includegraphics[scale=.6]{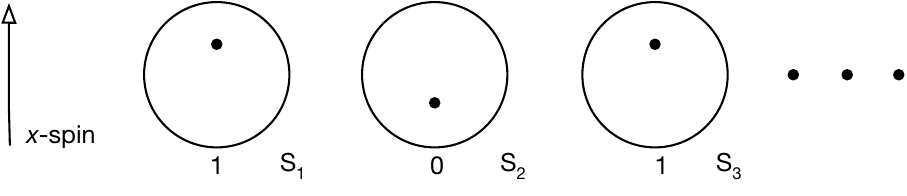}
    \caption{Particles distributed using the Martin-L\"{o}f random sequence $\sigma_{MLR}$.}
\end{figure}
And suppose God uses the following rule for distributing the particles at an initial time relative to $\sigma_{MLR}$: if the $k$th digit of the sequence is $0$, put $S_k$ in the bottom of its respective wave packet; else put it in the top.

On this construction, if we perform $x$-spin measurements on each, we will end up with a Martin-L\"{o}f random sequence of records
$$
r_{\text{out}}=(\uparrow_x,\downarrow_x,\uparrow_x, \ldots)
$$
with limiting relative frequency $1/2$ for each possible outcome. Here the sequence of results simply recapitulates $\sigma_{MLR}$. Since the sequence of measurement results is ML random, we get unbiased relative frequencies with certainty in the limit, not just probability~1.

While this way of determining the initial configuration will not give us the right statistics for other spin observables, it is easy to specify a construction that does.

Consider an unbiased ML random string
$$
\sigma_{MLR}=(1,0,0,1,0,1 \ldots)
$$
is ML-random, and suppose again that we have an infinite series of wave packets in eigenstates of $z$-spin. Here God uses $\sigma_{MLR}$ to place the particles recursively by interpreting $0$ to mean bottom and left and $1$ to mean top and right. Since the \textit{first term} in $\sigma_{MLR}$ is 1, particle $S_1$ is placed in the \textit{top half} of its wave packet. This is the first line of figure~3.
\begin{figure}[!ht]
  \centering
    \includegraphics[scale=.6]{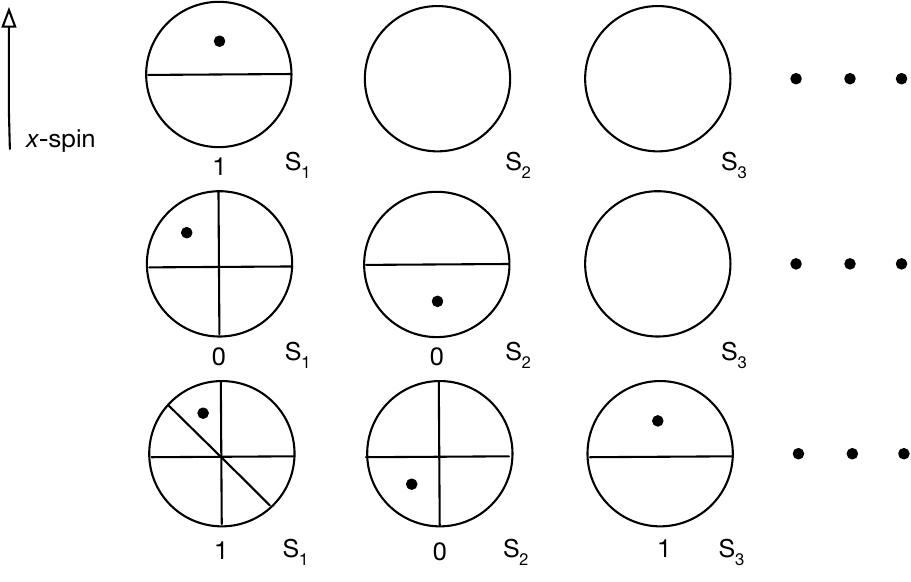}
    \caption{A slightly more sophisticated random distribution.}
\end{figure}
Then since the next \textit{two} terms in $\sigma_{MLR}$ are 0 and 0, particle $S_1$ is shifted to the \textit{left side} of the top half of its wave packet and particle $S_2$ is placed in the \textit{bottom half} of its wave packet. This is the second row of the figure. For the third row, since the next \textit{three} terms in $\sigma_{MLR}$ are 1, 0, and 1, particle $S_1$ is shifted to the \textit{right side} of the left side of the top half of its wave packet, $S_2$ is shifted to the \textit{left side} of the bottom half of its wave packet, and $S_3$ is placed in the \textit{top half} of its wave packet. And so on.

On this construction, the particles are randomly distributed in their wave packets uniformly with respect to angle. As a result, any particular spin observable in the $xy$-plane will result in an ML-random sequence of outcomes with the standard quantum limiting relative frequencies of $1/2$ for each possible outcome.

While this improved construction is suggestive, we will need something significantly more subtle to get the standard quantum statistics generally. There are at least three ways in which the present construction falls short of what one might want. Further, there is an issue regarding how the statistics connect to single-shot credences in algorithmic Bohmian mechanics.

First, on the present construction, there will always exist sequences of \textit{different} spin measurements where the results are sure to be nonrandom. Indeed, if one knew where the particles were in their respective wave packets, one could choose a sequence of spin measurements such that each spin result was \textit{up} in the direction measured for each particular particle.

Because of how we ultimately want to understand probability, we need to \textit{guarantee} that measurement results will be randomly distributed with the Born relative frequencies regardless of what spin observables one wishes to measure. While there is no such guarantee in standard Bohmian mechanics, a modest constraint on experimental protocols guarantees randomly distributed Born statistics in algorithmic Bohmian mechanics. The strategy is to restrict aBM to \textit{computable sequences of measurements}. Specifically, the theory represents sequences of measurements as layerwise computable functions relative to the standard Born measure. We will discuss what this means and motivate the representation when we have the requisite technical notions at hand.

Second, the present construction does not work for repeated measurements on a single system. Consider alternating $x$- and $y$-spin measurements of particle $S_1$. While the first $x$-spin measurement will yield the first term of $\sigma_{MLR}$ and the first $y$-spin measurement will yield the second term, the construction allows for a degree of freedom such that both $x$-spin and $y$-spin results will, contrary to the standard quantum statistics, exhibit biased relative frequencies in the limit.\footnote{Specifically, if one supposes that each particle is the same distance from the center of its wave function, one induces a measure in both the $x$- and $y$-directions that differs from the quantum probability measure represented by the wave function.} We return to this case in \S8.

Third, and closely related, while the present construction works fine for a sequence of idealized measurements of a particular spin observable in the $xy$-plane when each particle is associated with a uniform, symmetric, separable wave packet, one should want a construction that works when the initial wave function is neither uniform, symmetric, nor separable and that yields the right quantum statistics for other observables as well.

We will address these considerations by showing how to choose a \textit{point} in configuration space that is ML-random relative to an arbitrary computable measure $\mu_\psi$. The distribution postulate will then be the simple constraint that the initial configuration is an ML-random point relative to the initial Born measure $\mu_\psi$.

Extending the notion of Martin-L\"{o}f randomness from unbiased sequences to points in configuration space relative to a specified measure is mathematically straightforward but requires some conceptual care. We will show how to do this in the next section.

Finally, algorithmic Bohmian mechanics guarantees the standard quantum statistics for sequences of measurements like those discussed in this section, but the right statistics do not by themselves entail the right credences. While one might expect an inference from physical chances to credences to require something akin to the Principal Principle, because aBM \textit{guarantees} the right quantum statistics in the limit, one just needs exchangeable credences over measurement outcomes to get the standard Born probabilities as single-shot credences by means of Bruno de Finetti's representation theorem.\footnote{See de Finetti (1973-1974) for his views, Diaconis and Skyrms (2018) for a discussion of de Finetti's representation theorem, and Barrett and Chen (2025) and (2026) for discussions of $L^\star$ laws like the distribution postulate formulated here, the Principal Principle, and exchangeability.} We will return to this point briefly at the end of the paper.

\section{Computability of states and measurements}\label{sec:Computability of measurements}

One of the fundamental results in the study of Martin-Löf randomness is the existence of a universal test.
A \textit{universal Martin-L\"{o}f test} is a Martin-L\"{o}f test $\{U_n\}_{n \in \mathbb{N}}$ such that for every Martin-L\"{o}f test $\{V_n\}_{n \in \mathbb{N}}$ there exists $c \in \mathbb{N}$ such that $V_{n+c} \subseteq U_n$ for all $n$.
In other words, the universal test eventually contains \textit{any} test, and hence \textit{any} Martin-L\"{o}f null set. As a result, any Martin-L\"{o}f random point must fall outside the null set defined by the universal test; for this reason we can call that null set a \textit{universal} Martin-L\"{o}f null set.

We can use the notion of a universal Martin-L\"{o}f test to specify directly which \textit{points} in configuration space are Martin-L\"{o}f random relative to the Born measure $\mu_\psi$. Suppose $\mu_\psi$ is a computable measure on configuration space, a computable Polish space. A universal $\mu_\psi$–Martin–Löf test is a uniformly effectively open sequence of sets $\{U_n\}$ such that
\[
\mu_\psi(U_n) \;\le\; 2^{-n} \qquad \text{for all } n,
\]
and with the property that for any other $\mu_\psi$–Martin–Löf test
$\{V_n\}_{n\in\mathbb{N}}$ there exists a constant $c\in\mathbb{N}$
satisfying
\[
V_{n+c} \;\subseteq\; U_n
\qquad\text{for all } n.
\]
Intuitively, the sets $U_n$ identify those configurations that fail
increasingly stringent effective tests for being $\mu_\psi$–random. The set of $\mu_\psi$–Martin–Löf random points at time $t_0$ is
$$
\mathsf{MLR}^{\mu_\psi}
:=\;
X \setminus \bigcap_{n\in\mathbb{N}} U_n.
$$
where $X$ is configuration space and $\bigcap_{n} U_n$ are the effectively describable non-random points with respect to the measure $\mu_\psi$ determined by $|\psi(t_0)|^2$.

The set $\mathsf{MLR}^{\mu_\psi}$ is independent (up to finite changes of index) of the particular choice of universal test. It allows us to say what it means for the initial configuration to be typical relative to measure $\mu_\psi$. As a result, the algorithmic distribution postulate says just this: the initial configuration $x_0\in X$ is chosen so that
$$
x_0 \in \mathsf{MLR}^{\mu_\psi}.
$$

Now we turn to the problem of finding a suitable representation for a computable sequence of measurements. The key idea here is that of a layerwise computable function. To say what that is, we first need the notion of randomness deficiency.

The universal Martin-L\"{o}f test allows us to make distinctions among $\mu$-Martin-L\"{o}f random points.
Intuitively, some random points are ``more random'' than others. The idea is particularly clear in the context of random sequences.
Consider a Martin-L\"{o}f random sequence that begins with 1 million consecutive zeroes, and after that is ``random''.
The sequence obtained by deleting those zeroes is also Martin-L\"{o}f random, and moreover is intuitively ``more random''.
This intuition can be made precise via the \textit{randomness deficiency} of a point in configuration space $X$, which we now define.
\begin{definition}
    Fix a universal $\mu$-Martin-L\"{o}f test $\{U_n\}_{n \in \mathbb{N}}$.
    Let $\mathsf{MLR}^\mu_n := X \setminus U_n$.
    For any $\mu$-Martin-L\"{o}f random $x \in X$ the minimal $n$ such that $x \in \mathsf{MLR}^\mu_n$ is the \textit{randomness deficiency} of $x$.
\end{definition}
Here are a few important facts: $\mathsf{MLR}^\mu = \bigcup_n \mathsf{MLR}^\mu_n$; for all $n$, $\mu(\mathsf{MLR}^\mu_n) \geq 1-2^{-n}$ and $\mathsf{MLR}^\mu_n$ is effectively compact;\footnote{A set $K \subseteq X$ is \textit{effectively compact} if it is effectively closed and for any countable open cover $\{C_n\}_{n \in \mathbb{N}} \supseteq K$ there is a partial computable function that computes a finite subcover.}
and randomness deficiency can be defined relative to any notion of ``universal test'' (e.g., compressibility by universal Turing machine, or universal lower semicomputable martingale).
Intuitively those points that are more random have lower randomness deficiency; they are judged not to exhibit a computable pattern relative to the measure sooner by the fixed universal test and are hence ruled out sooner.

The core idea behind a \textit{layerwise computable} function is that it is computable on each ``layer'' $\mathsf{MLR}^\mu_n$ of the Martin-L\"{o}f randoms.
\begin{definition}
    Let $(X, \mu)$ be a computable probability space and let $Y$ be a computable Polish space.
    Then $f: X \to Y$ is $\mu$-\textit{layerwise computable} if for all $k \in \mathbb{N}$ and all $x \in \mathsf{MLR}^\mu_k$, $f(x)$ is computable uniformly in $x, k$.
\end{definition}
So given a Martin-L\"{o}f random point $x$ and its randomness deficiency $k$, there is a single machine that computes $f(x)$ on the basis of $x, k$.
Thus a sequence $\{f_n\}_{n \in \mathbb{N}}$ of functions is \textit{uniformly} $\mu$-\textit{layerwise computable} if there is a single machine which, given a random point $x$ and its randomness deficiency, computes each $f_i(x)$.

In particular, \textit{only} the Martin-Löf random points of $X$ have a finite randomness deficiency.
If $x \in X$ is not $\mu$-Martin-Löf random, then it is not in $\mathsf{MLR}^\mu_n$ for any $n$.
So another way to define a layerwise computable function is: $f$ is layerwise computable if and only if $f$ is computable \textit{on the Martin-Löf random points}.
The behavior of $f$ outside the Martin-Löf random points need not be computable;
but, since $\mathsf{MLR}^\mu$ is a $\mu$-measure one set, this means that layerwise computable functions are computable almost everywhere.
So they are a natural generalization of (everywhere) computable functions, which are better suited to our current measure-theoretic setting.\footnote{See the surveys by Hoyrup (2020) and Hoyrup and Rute (2021) for more details on layerwise computable functions, and arguments to the effect that they are the natural effective counterpart to measurable functions.}
Layerwise computable functions are also particularly nice because of the following result due to Hoyrup and Rojas (2009).
\begin{theorem}\label{thm:LayerwisePushforwardsPreserveMLR}
    Suppose $(X, \mu)$ is a computable probability space and $Y$ is a computable Polish space.
    Let $f: X \to Y$ be a $\mu$-layerwise computable function.
    Then the \textit{push-forward measure} defined
    \[\mu_f(A) = \mu(f^{-1}(A)),\]
    for $A \subseteq Y$ measurable, is itself a computable probability measure, and moreover if $x \in X$ is $\mu$-Martin-L\"{o}f random then $f(x)$ is $\mu_f$-Martin-L\"{o}f random.
\end{theorem}

\begin{definition}
    A \textit{computable binary measurement} is a $\mu_\psi$-layerwise computable random variable $f: X \to \{0,1\}$.
    A \textit{computable discrete measurement} is a $\mu_\psi$-layerwise computable random variable that takes values in $\mathbb{N}$.
    Finally, a \textit{computable measurement} is a $\mu_\psi$-layerwise computable \textit{real-valued} random variable.
\end{definition}

That is, a computable binary measurement takes the value 0 on some subset of configuration space and takes the value 1 on the complement of that subset. This is a general form of a computable spin-measurement, for example. A computable discrete measurement can take finite or countably many values. This corresponds to a more sophisticated measuring device, one which can split the wave function into a discrete partition. The most general definition accommodates real-valued position measurements. This is an idealization since real measuring devices have finite accuracy, but there's no \textit{mathematical} reason we can't work with them.

Consider how a $\mu_\psi$-layerwise-computable function
$$
f : X \to \{0,1\}
$$
can be used to represent a single $x$-spin measurement in Bohmian mechanics. A measurement in Bohmian mechanics involves the evolution of the composite system consisting of the object system and the measurement apparatus. The result of the measurement is determined by the pointer position which, in turn, is a function of the initial configuration of the composite system. There are two aspects to the measurement: (1) the deterministic flow of the configuration induced by the evolution of the wave function under the dynamics and (2) the partition of configuration space into regions corresponding to different final pointer positions. The $\mu_\psi$-layerwise computable function~$f$ represents the composition of the two.

With this in mind, let $X = \mathbb{R}^{3N}$ be the configuration space of the composite system: the particle being measured, the $x$-spin measuring device with its pointer, and the relevant environment. A single point $x_0 \in X$ gives the complete initial configuration of the composite system at time $t_0$. We will suppose that the measurement interaction occurs between times $t_0$ and $t_1$.

Given the global Hamiltonian, the Bohmian guiding equation generates a deterministic flow
$$
\Phi^{\psi_0}_{t_1} : X \to X,
$$
so that for initial configuration $x_0$ and wave function $\psi_0$,
$$
x(t_1) = \Phi^{\psi_0}_{t_1}(x_0)
$$
is the configuration after the interaction between the particle and apparatus has correlated the $x$-spin of the particle with the position of the pointer. If the wave function and Hamiltonian are layerwise computable, then the flow $\Phi^{\psi_0}_{t_1}$ is $\mu_\psi$-layerwise computable.

After the interaction and decoherence, the apparatus pointer occupies one of two macroscopically disjoint regions of configuration space. We will call these regions ``$x$-spin up'' and ``$x$-spin down'' indicating what the two pointer configurations mean. Define a partition map
$$
g : X \to \{0,1\}
$$
by
$$
g(x) =
\begin{cases}
1, & x \in \mbox{``$x$-spin up''}, \\[6pt]
0, & x \in \mbox{``$x$-spin down''}.
\end{cases}
$$
Then a single x-spin measurement is represented by the composition
$$
f(x_0) \;=\; g\bigl(\Phi^{\psi_0}_{t_1}(x_0)\bigr).
$$
In this way, the function $f:X \to \{0,1\}$ encodes the entire $x$-spin measurement interaction and maps initial configurations to the pointer outcomes that would be realized if those initial configurations obtained. The function $f$ is $\mu_\psi$-layerwise computable because both $\Phi^{\psi}_{t_1}$ and $g$ are $\mu_\psi$-layerwise computable on all Martin–Löf random initial configurations. Its being $\mu_\psi$-layerwise computable means that it is computable on all configurations that can occur if the initial configuration is $\mu_\psi$-Martin–Löf random. And that the initial configuration is in fact $\mu_\psi$-Martin–Löf random is ensured by the algorithmic distribution postulate.

In a way very like the $x$-spin measurement just described, a single $\mu_\psi$-layerwise computable function $f$ can represent any computable sequence of measurements and their outcomes. That this $f$ will be computable on all $\mu_\psi$-Martin-L\"{o}f random initial configurations is what makes $\mu_\psi$-layerwise computability the right notion for capturing what it means in the theory for a sequence of measurements to be computable.

While we have not appealed to the equivariance of the standard Born measure under the Bohmian dynamics, it plays an essential role in aBM. Since a measurement interaction is represented by a $\mu_\psi$-layerwise computable map $f$ expressible as a composition of the Bohmian flow and the pointer partition, it follows from the equivariance of the Born measure that
\begin{equation}\label{eq:EquivarianceTransport}
\mu_{\psi_t}(A)
\;=\;
\mu_{\psi_0}\bigl((\Phi^{\psi_0}_t)^{-1}(A)\bigr).
\end{equation}
for any measurable $A\subseteq X$ and every time $t$ that the push-forward measure induced by $f$ is the Born measure given the wave function at time $t$.

In the context of algorithmic Bohmian mechanics this means that the main theorem of the next section will entail that a $\mu_\psi$-layerwise computable map $f$ preserves Martin–Löf randomness relative to the time-evolved Born measure on possible configurations. Hence, a configuration that is typical relative to the Born measure at time $t_0$ will be mapped to a configuration that is typical relative the push-forward Born measure over pointer positions at time $t$. And the distribution postulate itself will ensure that the initial configuration is typical relative to the Born measure at time $t_0$.

\section{Martin-L\"{o}f random outcomes}

We now show that a $\mu_\psi$-layerwise computable sequence of measurements $\{f_n\}$ preserves algorithmic randomness by transforming a ML-random initial configuration in Born measure into a sequence of ML-random measurement outcomes in a push-forward measure that respects the standard Born statistics. Specifically, the theorem~2 below entails that if the initial configuration $x_0$ is Martin-L\"{o}f random with respect to $\mu_{\psi}$, the corresponding sequence of measurement outcomes $(f_1(x_0), f_2(x_0), \ldots)$ will be Martin-L\"{o}f random with respect to a push-forward joint measure $\mu_\omega$ induced by $\mu_{\psi}$ and the $\mu_{\psi}$-layerwise computable $\{f_n\}$.

Let $X$ be a computable Polish configuration space.
Let $\mu_\psi$ denote the quantum probabilities on $X$ induced by state $|\psi\rangle$.
Suppose $\{f_n\}$ is a uniformly computable sequence of measurements (i.e., a uniformly $\mu_\psi$-layerwise computable sequence of functions).
From each initial segment $\{f_i\}_{i \leq n}$ of the sequence we define a layerwise-computable function $f_{1:n}: X \to \mathbb{R}^n$ defined by $f_{1:n}(x)=(f_1(x), f_2(x), \ldots, f_n(x))$.
Each $f_{1:n}$ induces a joint probability measure $\mu_n$ on $\mathbb{R}^n$ given by 
\[\mu_n(A)=\mu_{\psi}(f_{1:n}^{-1}(A)).\]

By Kolmogorov's extension theorem the family $\{\mu_i\}_{i \in \mathbb{N}}$ extends to a unique joint measure $\mu_\omega$ on $\mathbb{R}^\omega$.
We want to show that for any $x$ that is Martin-L\"{o}f random with respect to $\mu_{\psi}$ the corresponding outcome sequence $(f_1(x), f_2(x), \ldots)$ is Martin-L\"{o}f random in $\mathbb{R}^\omega$ with respect to the joint measure $\mu_\omega$.

\begin{lemma}\label{lemma:LayerwiseComputableMeasurementSequence}
    Let $(X, \mu)$ be a computable probability space.
    Suppose $\{f_i\}_{i \in \mathbb{N}}$ is a uniformly computable sequence of measurements.
    Then the function $f_\omega: X \to \mathbb{R}^\omega$ defined
    \[f_\omega(x) = (f_1(x), f_2(x), \ldots)\]
    is $\mu$-layerwise computable.
\end{lemma}

\begin{proof}
    This follows from the definition of uniform layerwise computability: given $x \in \mathsf{MLR}^\mu_k$ each $f_i(x)$ is computable uniformly from $x, k$, and these values determine $f_\omega(x)$.
\end{proof}
In particular this implies that $\mu_\omega$ is itself computable, since it is the probability measure induced by $f_\omega$.

\begin{theorem}\label{thm:MeasurementOutcomesAreMLR}
    Let $(X, \mu_\psi)$ be the configuration space of a quantum system together with the quantum probability measure induced by the wave function $\psi$.
    Suppose $\{f_n\}$ is a uniformly computable sequence of measurements on $X$, $f_\omega$ is the induced product function, and $\mu_\omega$ is the induced joint probability measure on the outcomes of those measurements.
    Then for any $x \in X$ that is $\mu_\psi$-Martin-L\"{o}f random, $(f_1(x), f_2(x), \ldots) \in \mathbb{R}^\omega$ is $\mu_\omega$-Martin-L\"{o}f random.
\end{theorem}

\begin{proof}
    By Lemma~\ref{lemma:LayerwiseComputableMeasurementSequence}, $f_\omega$ is $\mu_\psi$-layerwise computable.
    By Theorem~\ref{thm:LayerwisePushforwardsPreserveMLR} $f_\omega$ therefore preserves Martin-L\"{o}f randoms.
\end{proof}

\section{An algorithmic formulation of the distribution postulate}

Putting the pieces together, we will state the algorithmic distribution postulate, then show that aBM guarantees the standard Born statistics for the sort of experiments we considered at the beginning of the paper.

Let $X = \mathbb{R}^{3N}$ be configuration space, and let
\[
\mu_\psi(A) = \int_A |\psi(x)|^2 \, d\lambda(x)
\]
be the Born measure induced by a computable wave function $\psi$. Assuming $(X,\mu_\psi)$ is a computable probability space, the distribution postulate becomes:

\begin{quote}
\textbf{Algorithmic Distribution Postulate:}
The actual initial configuration $x_0 \in X$ is \emph{Martin-L\"{o}f random} with respect to the Born measure $\mu_\psi$.
\end{quote}
That is, God simply chooses an initial configuration that lies outside every effectively describable $\mu_\psi$-null set. This is a perfectly clear objective constraining law. It says that the initial configuration exhibits \textit{no effectively describable regularities} relative to the quantum measure $\mu_\psi$. In other words, the initial configuration is \textit{typical in every computable sense} relative to the standard Born measure.

If the initial configuration $x_0 \in X$ is Martin-L\"{o}f random with respect to $\mu_\psi$, then a layerwise computable sequence of measurements will yield a Martin-L\"{o}f random sequence of outcomes relative to the push-forward quantum statistical measure $\mu_\omega$. Specifically, in the case of arbitrary computable sequences of spin measurements on identically prepared, non-interacting systems, one can show that we will get a ML-random sequence of results that exhibits the standard limiting Born statistics.

\begin{corollary}\label{cor:LimRelFrequencies}
    Let $(X, \mu_\psi)$ be the computable probability space induced by the initial quantum state $\psi$ and suppose that $\{f_i\}_{i \in \mathbb{N}}$ is a uniformly layerwise-computable sequence of measurements of identically prepared, non-interacting systems. Let $f_{1:n}: X \to \mathbb{R}^n$ denote the $n$-fold product of the first $n$ measurements.
    If $x \in X$ is Martin-L\"{o}f random then
    \[\lim_{n \to \infty}\frac{1}{n}f_{1:n}(x) = \mathbb{E}_{\mu_\psi}[f_1].\]
\end{corollary}

\begin{proof}
    By Theorem~\ref{thm:MeasurementOutcomesAreMLR} we know that $f_{1:n}(x)$ is the first $n$ coordinates of a Martin-L\"{o}f random vector in $\mathbb{R}^n$.
    The corollary then follows from the fact that Martin-L\"{o}f random sequences satisfy the law of large numbers.
\end{proof}

Two small notes regarding Corollary~\ref{cor:LimRelFrequencies} and its interpretation. First, if $\psi$ is a computable wave function (in the sense of computable analysis), then the Born measure $\mu_\psi$ is a computable probability measure on $X$. And since we are supposing identically prepared, non-interacting systems, this gives us independence. So $\mu_\omega$ is a computable product measure, and the strong law of large numbers applies.

As a particular case of Corollary~\ref{cor:LimRelFrequencies}, suppose that every $f_i$ takes values in a finite set $V$, as would be the case in a sequence of spin measurements.
Then we have that for each value $v \in V$ as $n \to \infty$,
\[\frac{\#v(f_{1:n}(x))}{n} \to \mu_f(v),\]
that is, the outcome sequence exhibits the correct limiting relative frequencies.

More generally, if one measures a particular observable, say $x$-spin, on any computable subsequence of the systems, then the results will be Martin-L\"{o}f random with the standard quantum limiting relative frequencies since every infinite computable subsequence is itself Martin-L\"{o}f random with the same relative frequencies. As a result, every computable subsequence of measurements will yield the standard quantum statistics in the limit for every observable that one tries infinitely often.

Hence, if the initial configuration $x_0$ is $\mu_\psi$-Martin-L\"{o}f random, the results of any layerwise computable sequence of spin measurements of identically prepared, non-interacting systems will certainly be Martin-L\"{o}f random with the standard Born relative frequencies in the limit. This is the strong guarantee we wanted from the algorithmic theory. It means that one will get empirical results that confirm the theory in every physically possible world.\footnote{See Barrett and Chen (2025) and (2026) for discussions of this point and the virtues of such behavior.}

One will also certainly get Martin-L\"{o}f random outcomes that are distributed according to the standard quantum statistics for sequences of measurements of a single system. The example of alternating $x$- and $y$-spins on a single particle $S$ that we considered in \S5 (the second of the three shortcomings of the toy constructions) illustrates this.

\begin{corollary}
    Let $(X, \mu_\psi)$ be the computable probability space induced by the initial state $\psi$ and suppose that $\{g_i\}_{i \in \mathbb{N}}$ is a uniformly layerwise-computable sequence of alternating $x$-spin and $y$-spin measurements of a single particle $S$.
    As before let $g_{1:n}:X \to \mathbb{R}^n$ be the $n$-fold product of the first $n$ measurements.
    If $x \in X$ is $\mu_{\psi}$-Martin-Löf random then
    \[\lim_{n \to \infty}\frac{1}{n}g_{1:n}(x)=\frac{1}{2}.\]
\end{corollary}

\begin{proof}
    Since the full sequence of measurements $g_i$ consist of alternating splittings of the wave function along the $x$ and $y$ axes, the sequence is independent and identically distributed under $\mu_\psi$.
    Hence, the limiting relative frequencies of up and down deflections for both the $x$- and $y$-spin measurements is $1/2$.
    As in the proof of Corollary \ref{cor:LimRelFrequencies}, we know that $\lim_{n \to \infty}g_{1:n}(x)$ is Martin-Löf random in the push-forward joint measure $\mu_\omega$ induced by $\mu_\psi$. 
    So, again by the law of large numbers, 
    \[\lim_{n \to \infty}\frac{1}{n}g_{1:n}(x)=\frac{1}{2}.\qedhere\]
\end{proof}

In each of these cases, choosing the initial configuration in accord with the algorithmic distribution postulate guarantees measurement records that are Martin-L\"{o}f random with relative frequencies given by the standard Born limiting relative frequencies for the relevant sequence of experiments. This addresses the experiments and issues we considered in section~5.

\section{Discussion}

We have shown that for layerwise computable interactions, the configuration being Martin-L\"of random relative to $\mu_\psi$ is preserved under the dynamics in algorithmic Bohmian mechanics (aBM), just as the epistemic probability density of the configuration $\mu_\psi$ is preserved under the dynamics in standard Bohmian mechanics. And, with respect to empirical predictions, we have shown how aBM guarantees a Martin-L\"of random sequence of measurement results that satisfies the standard Born statistics for each type of experiment described in Section~5.

This does not, however, mean that the two theories make the same empirical predictions. Most saliently, aBM makes statistical predictions, not probabilistic predictions. And its statistical predictions are, in many cases, stronger than those of standard Bohmian mechanics, as it guarantees that the results of layerwise computable measurements will be Martin-L\"of random with the Born relative frequencies rather than merely predicting this with probability one.

While the Born statistics do not by themselves entail the standard forward-looking quantum probabilities, if (1) an agent’s credences are \textit{exchangeable} over the outcomes of a given sequence of measurements (i.e., if she believes that the order of the measurements does not matter), and (2) she accepts algorithmic Bohmian mechanics, then de Finetti’s representation theorem commits her to assigning the standard forward-looking Born credences for single-case outcomes to a sequence of spin measurements.\footnote{See Barrett and Chen (2026) for a detailed discussion of this point.} Thus, an agent with exchangeable credences is committed to assigning the standard single-shot credences to measurement outcomes for these experiments.

The upshot is that whenever the physical situation is such that one’s credences over outcomes are exchangeable, aBM commits one to the same single-shot credences as standard quantum mechanics. The experiments we have considered are cases in which one would not expect the order of the measurements to matter.\footnote{In the experiment where one alternates between $x$-spin and $y$-spin measurements, one would not expect the order of measurement pairs to matter.}

There remains work to be done in analyzing the theory. A significant open question concerns when Martin-L\"{o}f randomness relative to the conditional measure associated with a subsystem is preserved under the dynamics. This is a delicate topic for future study.

For the present, we have shown how to replace an informal appeal to typicality with a mathematically precise objective constraint in a way that guarantees the standard quantum statistical predictions for the experiments considered here. The algorithmic distribution postulate simply says that the initial configuration $x_0$ is in the set $\mathsf{MLR}^{\mu_\psi}$. And it is perfectly clear what this means.

\appendix

\section{Algorithmic Randomness and Computable Measure Theory}

Much of the recent work on algorithmic randomness has involved extending randomness notions---e.g., Martin-L\"{o}f randomness---to spaces other than $2^\omega$.\footnote{See e.g. Hoyrup and Rojas (2009), Rute (2013, 2020).}
This is important because $2^\omega$ has a uniquely well-behaved topology with regard to computability theory.
In particular the canonical topology on $2^\omega$ is generated by the countable collection of \textit{cones} (or \textit{cylinder sets}) $[\sigma] := \{S \in 2^\omega \mid \sigma \prec S\}$, where $\sigma \in 2^{<\omega}$ and we write ``$\sigma \prec S$'' to indicate that $\sigma$ is a proper initial segment of $S$.
These cones are clopen in the canonical topology; in the arithmetical hierarchy they are computable ($\Delta^0_1$) sets.
A well-known result of general topology states that a topological space $\Omega$ has nontrivial clopen subsets if and only if $\Omega$ is totally disconnected; thus important spaces such as $\mathbb{R}, \mathbb{C}, C[0,1]$, and $L^p$ spaces do not have nontrivial clopen subsets.
We therefore work with a more general notion of computability on topological spaces, one that has been thoroughly studied in the computable analysis literature.\footnote{See for example Brattka and Hertling and Weihrauch (2008).}

\begin{definition}
    A \textit{computable Polish space} is a triple $(X, \pazocal{S}, d)$ such that
    \begin{itemize}
        \item $(X, d)$ is a separable complete metric space;
        \item $\pazocal{S} \subseteq X$ is countable and dense; 
        \item for all $s_0, s_1 \in \pazocal{S}$, $d(s_0, s_1)$ is a uniformly computable real.
    \end{itemize}
\end{definition}
We will often simply say that $X$ is a computable Polish space, rather than refer to the entire triple.
Computable Polish spaces have sufficient structure to define complexity notions for their subsets via the arithmetic hierarchy, as we now show.
As usual we write $B(x,r)$ to denote the basic open ball centered at $x \in X$ with radius $r \in \mathbb{R}$.
\begin{definition}
    Let $X$ be a computable Polish space.
    A set $U \subseteq X$ is \textit{effectively open} if it is of the form
    \[U = \bigcup_{i \in I}B(s_i, q_i)\]
    where $I \subseteq \mathbb{N}$ is c.e. and for all $i$, $s_i \in \pazocal{S}$ and $q_i \in \mathbb{Q}$. 
\end{definition}
Effectively open sets are precisely the $\Sigma^0_1$ subsets.
We can therefore define $\Pi^0_1$ sets (which in computable analysis are called ``effectively closed sets'') as complements of effectively open sets, $\Sigma^0_2$ sets as effective countable unions of effectively closed sets, etc.
We then lift the computability notions from sets to higher structures, starting with functions.
\begin{definition}
    Let $X$ and $Y$ be computable Polish spaces.
    A function $f:X \to Y$ is \textit{computable continuous} (or sometimes simply \textit{computable}) if for every effectively open $U \subseteq Y$, $f^{-1}(U) \subseteq X$ is uniformly effectively open.
\end{definition}
From the definition it is clear that computability is a strong form of continuity, which is a basic fact of computable analysis (see Weihrauch 2000).

Finally we define computability for measures on computable Polish spaces.
\begin{definition}
    Let $X$ be a computable Polish space and let $\mu$ be a $\sigma$-finite Borel measure on $X$.
    We say that $\mu$ is \textit{computable} if for every effectively open $U \subseteq X$, either $\mu(U) = \infty$ or $\mu(U)$ is a left-c.e. real uniformly in (a code for) $U$.
\end{definition}
It follows from the definition that if $\mu$ is a finite measure (in particular, a probability measure) then for all effectively open sets $U$, $\mu(U)$ is a uniformly left-c.e. real.
When $X$ is a computable Polish space and $\mu$ is a computable probability measure on $X$, we will call the pair $(X, \mu)$ a \textit{computable probability space.}

Computable probability theory has developed into a rich, well-studied domain (see Rute (2020) and Hoyrup and Rute (2021) for overviews).
In particular, computable probability measures on computable Polish spaces are the right structures to generalize algorithmic randomness notions.
\begin{definition}
    Let $(X, \mu)$ be a computable probability space.
    A $\mu$-\textit{Martin-L\"{o}f test} is a sequence of uniformly effectively open sets $\{U_n\}_{n \in \mathbb{N}}$ such that $\mu(U_n) \leq 2^{-n}$.
    A point $x \in X$ is $\mu$-\textit{Martin-L\"{o}f random} if for any $\mu$-Martin-L\"{o}f test $\{U_n\}_{n \in \mathbb{N}}$, $x \notin \bigcap_{n=1}^\infty U_n$.
\end{definition}
So we recover a suitably generalized definition of Martin-Löf randomness, one that applies to any computable probability measure over any space with sufficiently rich structure (specifically, any computable Polish space).

\begin{center}
\large{Bibliography}
\end{center}

\vspace{.25cm}
\noindent
Albert, David Z.\ (1992) \textit{Quantum Mechanics and Experience}. Cambridge, MA: Harvard University Press.

\vspace{.25cm}
\noindent
Allori, Valia (2020)  (ed.), \textit{Statistical Mechanics and Scientific Explanation}. Singapore: World Scientific.

\vspace{.25cm}
\noindent
Barrett, Jeffrey A.\ (2019) \textit{The Conceptual foundations of Quantum Mechanics}. Oxford: Oxford University Press. 

\vspace{.25cm}
\noindent
Barrett, Jeffrey A.\ and Chen, Eddy Keming. (2025) ``Algorithmic Randomness and Probabilistic Laws,'' \emph{The British Journal for the Philosophy of Science}, forthcoming.  \\ \href{https://arxiv.org/abs/2303.01411}{https://arxiv.org/abs/2303.01411}

\vspace{.25cm}
\noindent
Barrett, Jeffrey A.\ and Chen, Eddy Keming. (2026) ``Deriving the Principal Principle from Algorithmic Randomness and Exchangeability,'' manuscript.

\vspace{.25cm}
\noindent
Barrett, Jeffrey A.\ and Simon Huttegger (2021)  ``Quantum Randomness and Underdetermination,'' \textit{Philosophy of Science}, Volume 87, Issue 3.

\vspace{.25cm}
\noindent
Bell, John S.\ (1987) \emph{Speakable and Unspeakable in Quantum Mechanics}. Cambridge: Cambridge University Press.

\vspace{.25cm}
\noindent
Brattka, Vasco and Peter Hertling and Klaus Weihrauch (2008) ``A Tutorial on Computable Analysis.'' In: Cooper, S.B., Löwe, B., Sorbi, A. (eds) \textit{New Computational Paradigms}. New York: Springer.

\vspace{.25cm}
\noindent
Callender, Craig (2007) ``The Emergence and Interpretation of Probability in Bohmian Mechanics,'' \emph{Studies in History and Philosophy of Modern Physics} \textbf{38}(2): 351--370.

\vspace{.25cm}
\noindent
Dasgupta, Abhijit (2011) ``Mathematical Foundations of Randomness,'' in Prasanta Bandyopadhyay and Malcolm Forster (eds.), \emph{Philosophy of Statistics (Handbook of the Philosophy of Science: Volume 7)}, Amsterdam: Elsevier, pp. 641–710. 

\vspace{.25cm}
\noindent
de Finetti, Bruno. (1974–1975) \emph{Theory of Probability}, Vols. 1 and 2. Trans. A. Machi and A. F. M. Smith. New York: Wiley.

\vspace{.25cm}
\noindent
Diaconis, Persi and Brian Skyrms (2018) \textit{Ten Great Ideas about Chance}, Princeton University Press: Princeton and Oxford.

\vspace{.25cm}
\noindent
D\"urr, Detlef, Sheldon Goldstein, and Nino Zangh\'i (1992) ``Quantum Equilibrium and the Origin of Absolute Uncertainty,'' \emph{Journal of Statistical Physics} \textbf{67}(5--6): 843--907.

\vspace{.25cm}
\noindent
D\"urr, Detlef, Sheldon Goldstein, and Nino Zangh\'i (1993) ``A Global Equilibrium as the Foundation of Quantum Randomness,'' \emph{Foundations of Physics} \textbf{23}(5): 721--738.

\vspace{.25cm}
\noindent
Goldstein, Sheldon (2012) ``Typicality and Notions of Probability in Physics,'' in Y.\ Ben-Menahem and M.\ Hemmo (eds.), \emph{Probability in Physics}, Berlin: Springer, pp.~59--71.

\vspace{.25cm}
\noindent
Goldstein, Sheldon (2025) ``Bohmian Mechanics,'' \emph{The Stanford Encyclopedia of Philosophy}, (Fall 2025 Edition), Edward N. Zalta \& Uri Nodelman (eds.), \\ \href{https://plato.stanford.edu/archives/fall2025/entries/qm-bohm/}{https://plato.stanford.edu/archives/fall2025/entries/qm-bohm/}

\vspace{.25cm}
\noindent
H\'ajek, Alan (2023) "Interpretations of Probability", in Edward N. Zalta \& Uri Nodelman (eds.) \emph{The Stanford Encyclopedia of Philosophy (Winter 2023 Edition)}. 

\noindent
\href{https://plato.stanford.edu/archives/win2023/entries/probability-interpret/}{https://plato.stanford.edu/archives/win2023/entries/probability-interpret/}

\vspace{.25cm}
\noindent
Hoyrup, Mathieu. ``Algorithmic randomness and layerwise computability.'' Algorithmic Randomness: Progress and Prospects, edited by Johanna NY Franklin and Christopher P. Porter, Lecture Notes in Logic, Cambridge University Press, Cambridge (2020): 115-133.

\vspace{.25cm}
\noindent
Hoyrup, M., and Rojas, C. (2009). Applications of effective probability theory to Martin-Löf randomness. In International colloquium on automata, languages, and programming (pp. 549-561). Berlin, Heidelberg: Springer Berlin Heidelberg.

\vspace{.25cm}
\noindent
Hoyrup, M., and Rute, J. (2021). Computable measure theory and algorithmic randomness. In Handbook of Computability and Complexity in Analysis (pp. 227-270). Cham: Springer International Publishing.

\vspace{.25cm}
\noindent
Lazarovici, Dustin. (2023). \emph{Typicality Reasoning in Probability, Physics, and Metaphysics.} Palgrave Macmillan.

\vspace{.25cm}
\noindent
Lewis, David. (1980) ``A Subjectivist’s Guide to Objective Chance,'' in Richard C. Jeffrey (ed.), \emph{Studies in Inductive Logic and Probability}, Vol. II, Berkeley: University of California Press, 263–293.

\vspace{.25cm}
\noindent
Li, Ming and Paul Vit\'anyi (2019). \emph{An Introduction to Kolmogorov Complexity and Its Applications, 4th Edition}, New York, NY: Springer. 

\vspace{.25cm}
\noindent
Martin-L\"of, Per (1966) ``The definition of random sequences,'' \emph{Information and Control} 9(6):602--619.

\vspace{.25cm}
\noindent
Maudlin, Tim. (2011) ``Three Roads to Objective Probability.'' In: C. Beisbart and S. Hartmann
(eds.) \emph{Probabilities in Physics}, pp. 293–319. New York, NY: Oxford University Press.

\vspace{.25cm}
\noindent
Miyabe, Kenshi. (2013). $L^1$-computability, layerwise computability and Solovay reducibility. Computability, 2(1), 15-29.

\vspace{.25cm}
\noindent
Porter, Christopher P. (2020) ``Biased Algorithmic Randomness'' in \emph{Algorithmic Randomness: Progress and Prospects},
Edited by Johanna N.\ Y.\ Franklin and Christopher P.\ Porter. Lecture Notes in Logic, 50, Association for Symbolic Logic. Cambridge: Cambridge University Press.

\vspace{.25cm}
\noindent
Rute, Jason (2013) "Topics in Algorithmic Randomness and Computable Analysis." Carnegie Mellon University. 

\vspace{.25cm}
\noindent
Rute, Jason (2020) ``Algorithmic Randomness and Constructive/Computable Measure Theory'' in \textit{Algorithmic Randomness: Progress and Prospects}. Ed. Johanna N. Y. Franklin and Christopher P. Porter. Lecture Notes in Logic, 50, Association for Symbolic Logic. Cambridge: Cambridge University Press, 58–114.

\vspace{.25cm}
\noindent
Valentini, Antony (1991) ``Signal-Locality, Uncertainty, and the Subquantum $H$-Theorem. I,'' \emph{Physics Letters A} \textbf{156}(1--2): 5--11.

\vspace{.25cm}
\noindent
Valentini, Antony and Hans Westman (2005) ``Dynamical Origin of Quantum Probabilities,'' \emph{Proceedings of the Royal Society A} \textbf{461}(2053): 253--272.

\vspace{.25cm}
\noindent
Weihrauch, Klaus (2000) \textit{Computable Analysis: An Introduction}. Springer-Verlag.

\vspace{.25cm}
\noindent
Wilhelm, Isaac. (2022). ``Typical: A theory of typicality and typicality explanation.'' \textit{British Journal for the Philosophy of Science}, 73:2. 

\end{document}